\documentclass[12pt,onecolumn]{IEEEtran}
\usepackage{graphicx}
\usepackage{amsmath}
\usepackage{amssymb}
\usepackage{graphicx}
\usepackage{amsthm}
\usepackage{array}
\usepackage{latexsym}
\usepackage{subfigure}
\usepackage{multicol}
\usepackage{mathtools,cuted}
\usepackage{stfloats}
\usepackage{float}
\usepackage{adjustbox}

\theoremstyle{definition}

\newtheorem*{que*}{Question}

\newtheorem{theorem}{Theorem}
\newtheorem{cor}{Corollary}
\newtheorem{definition}{Definition}

\newtheorem{lemma}{Lemma}
\newtheorem{lem}{Lemma}

\newtheorem{remark}{Remark}

\newtheorem{problem}{Problem}
\newtheorem{assumption}{Assumption}
\newcommand{\beq}{\begin{equation}}
\newcommand{\eeq}{\end{equation}}
\newcommand{\barr}{\left[\begin{array}}
\newcommand{\earr}{\end{array}\right]}

\newcommand{\bi}{\begin{itemize}}
\newcommand{\ei}{\end{itemize}}
\newcommand{\bnum}{\begin{enumerate}}
\newcommand{\enum}{\end{enumerate}}
\newcommand{\bc}{\begin{center}}

\newcommand{\stabS}{\mathbb{S}}

\begin{document}

\title{Robust stabilization of multiport networks}
\author{\IEEEauthorblockN{Mayuresh~Bakshi\IEEEauthorrefmark{2}~\IEEEmembership{Member,~IEEE,}, \thanks{\IEEEauthorrefmark{2}bmayuresh@ee.iitb.ac.in, Assistant Professor, Dept.\ of Engineering and Applied Sciences, VIIT Pune, India}
Virendra R Sule\IEEEauthorrefmark{1} \thanks{\IEEEauthorrefmark{1}vrs@ee.iitb.ac.in, Professor, Dept.\ of Electrical Engineering, IIT Bombay, India} \\ and Maryam Shojaei Baghini\IEEEauthorrefmark{3},~\IEEEmembership{Senior Member,~IEEE},\thanks{\IEEEauthorrefmark{3}mshoejai@ee.iitb.ac.in, Professor, Dept.\ of Electrical Engineering, IIT Bombay, India} }}

\maketitle

\begin{abstract}

This paper formulates and solves the problem of robust compensation of multiport active network. This is an important engineering problem as networks designed differ in parameter values due to tolerance during manufacture from their actual realizations in chips and hardware. Parameters also undergo changes due to environmental factors. Hence, practical use of networks requires compensation which is only possible by connecting compensating network at the ports. The resulting interconnection is then required to be stable over a range of parameter values. This is called robust compensation. This paper formulates such a problem using an extension of the coprime factorization theory well known in feedback control theory to the situation of multiport network interconnection developed in \cite{msm1} and formulates the robust stabilization problem as an $H_{\infty}$ optimization problem. The port interconnection of networks does not confirm with computation of the function of the interconnected network analogous to that of the feedback interconnection using signal flow graph. Hence the well known stabilization and stability theory of feedback systems cannot be utilized for such a problem. A new formulation of stabilization theory of network interconnection was formulated and developed by the authors in \cite{msm1}. The variations of parameters of the network are used to define a worst case neighborhood of the network in terms of its coprime fractions at the nominal values of parameters. The solution of the optimization problem is then carried out by the standard procedure of converting such a problem to the Nehari optimization problem \cite{fran}. This methodology of solving the robust compensation of multiport networks using feedback control theory is believed to be novel. \end{abstract}

\begin{IEEEkeywords}
Active networks, Coprime factorization, Feedback stabilization, Multiport network connections, Robust stabilization.
\end{IEEEkeywords}

\section{Introduction}
\IEEEPARstart {R}{obust} stabilization problem is one of the most practical problems of engineering. Engineering dynamical systems in practice almost always have uncertainty about values of parameters of their models or also have slow variations. Moreover such models also neglect dynamics at high frequencies causing fluctuation in a nominal model. Robust stabilization is concerned with designing feedback controllers of such plants so that the resultant dynamical system is stable in the entire worst case range of parameters and model fluctuations. Such a problem was definitively solved in the 1980s in feedback control theory by Kimura \cite{kimu}. The factorization approach to systems \cite{vids}\footnote{It is important to realize that the classical feedback control techniques of loop shaping only allowed stability analysis in terms of gain and phase margin with respect to known feedback. These techniques could not solve the problem of existence and design of controllers for given extent of variation in the model. Hence the solution obtained using factorization theory was a major theoretical advance and had great practical relevance.} 
provided systematic approaches to solving robust stabilization problem using the $H_{\infty}$ norm minimization. 


\subsection{Previous work and contribution of this paper}
In this paper we propose another novel application of the factorization approach and the $H_{\infty}$ optimization, to robust stabilization of active analog, multiport electrical circuits by port compensations. Such a solution, to the author's knowledge has never been proposed in circuit design or multiport network theory and constitutes solution of this important engineering problem which is novel for both control theory as well as network theory. However it must be pointed out that the importance of robust stabilization and analysis of robust stability in terms of feedback control theory of VLSI networks was proposed in \cite{waima}. This is paper is of great importance to application of feedback stability theory to VLSI networks but stops short of solving the stabilization (synthesis) problem while solving the robust stability (analysis) problem. Hence our paper can be considered as an advancement of the theory proposed by \cite{waima}. However, there are major technical differences in formulation of the network interconnection model and mathematical format of our paper relative to that in \cite{waima}. Our paper develops the stabilization theory in terms of port model of the interconnected network which does not utilize the feedback signal flow graph model. Moreover, we make use of the stable coprime factorization theory which is well known for its computational advantages over the polynomial coprime factorization theory used in \cite{waima}. Contributions in this paper can be understood as next steps than achieved in \cite{waima} on the issue of robust stability and stabilization problems of VLSI synthesis.   

To appreciate the nature and importance of this problem it is necessary to consider the situation of multiport compensation in networks analogous to feedback control in systems theory. The robust stabilization is then analogous to compensation with robust stability of interconnected network. Active multiport networks have uncertainties of parameters from their design values, slow variations in parameters due to aging, modeling errors due to neglected structures or stray capacitances etc. A mathematical model of a network at the nominal parameter values is usually known and a worst case percentage variation (tolerance) in parameters may be prescribed after manufacture. Even if such a network is designed for prescribed performance specifications, the performance and stability also undergo variation with changes in the physical network due to variations. It is thus an import engineering challenge to compensate the network at the ports by another multiport network. However, unlike in the case of passive networks, such interconnections are not guaranteed to be stable. Hence designing a stable interconnection of linear time invariant (LTI) multiport networks (which includes active networks) is an important problem which we term as the \emph{robust stabilization problem of multiport networks}.  

\subsection{Stabilization in network interconnection a new problem}
Question then arises what new difficulties such a problem posed to direct application of feedback stabilization theory ? The answer lies in the difference in the nature of network interconnection as compared to feedback systems. In multiport networks, the interconnection is in terms of physical voltage and current quantities at the ports. Here algebraic sums of variables occur for physically same type (voltage or current) variables and must follow Kirchhoff's laws at ports. These are not mathematical identifications and summations of variables as in signal flow graphs. Essentially, the connection between plant and controller in feedback systems occurs at two distinct signal points at the input and output. On the other hand port interconnection takes place at the same signal point the input being the source and output being the response of the same source. Due to this physical nature of control by port compensation, the traditional two input two output formulation of closed loop stability cannot be applied to stability of interconnection of networks. Even if one seeks to carry out such a formulation it turns out that the feedback signal flow graph of the interconnected network is extremely inconvenient to determine for multiport networks. Hence it is necessary to pose the stabilization problem of multiport interconnections from a fresh point of view. In 
\cite{msm1} authors realized that the feedback signal flow graph was not necessary to formulate this stabilization problem and proposed an approach using factorization theory to multiport feedback stabilization. A detailed account of this theory is reported in \cite{msm1}. Recently, the Bode sensitivity concept of feedback control \cite{horo,doft} is also extended to multiport compensation in \cite{msm2} as a first successful application of this approach using $H_{\infty}$ approach to sensitivity optimization by port compensation in networks. In the present paper we propose the second important engineering application of this stabilization theory of networks viz. to the robust stabilization by compensation of a network at the ports when the network has prescribed extent of variation of the model in frequency domain.    

\subsection{Background: Stable factorizations and multiport networks}
We shall extensively refer to background of stable factorization approach for LTI systems as in \cite{antm, vids, doft} and network theory from \cite{chdk}. The ring of stable proper rational functions is denoted as $\stabS$. A multiport network has each port indexed and a source of voltage or current type to be connected at each of the ports is fixed. Let the vector of sources defined as functions of time over $[0,\infty)$ at the ports be denoted $u_{s}$ and the vector of responses at the ports of same indices be denoted as $y_{r}$ over the same interval. The driving point function of the network is then the matrix $H(s)$ of rational functions with real coefficients which relates the vectors of Laplace transforms $U_{s}(s)$, $Y_{r}(s)$ of sources $u_{s}$ and responses $y_{r}$ when the network has zero internal initial conditions of its capacitor voltages and inductor currents, as
\beq\label{networkfunction}
Y_{r}(s)=H(s)U_{s}(s)
\eeq
Thus $H(s)$ is always a square matrix. A basic assumption about the matrix function $H(s)$ we shall make is that the network response at the ports matches that obtained from a doubly coprime fractional representation of $H(s)$ as in (\ref{dcfrelationwithTc}) below. Such an assumption is an algebraic equivalent of the dynamical assumption that the network has no unstable hidden modes. Unstable hidden modes for single port networks are discussed in \cite{chdk}. For multiport networks our description using doubly coprime fractions of $H(s)$ is a generalization.

To define the stabilization problem of network interconnection one first needs to define the notion of stability of a multiport network. This definition is reproduced from \cite{msm1} called as \emph{Bounded Source Bounded Response} stability analogous to the bounded input bounded output (BIBO) stability of LTI systems.

\begin{definition}[BSBR stability]
A multiport network with the vector of port sources denoted $u_{s}$ and the vector of port responses (in same order of indices of ports) denoted $y_{r}$ is BSBR stable if, when the initial voltages and currents in the capacitors and inductors inside the network are zero, then for any uniformly bounded sources $u_{s}$ in time $t\geq 0$ the responses $y_{r}$ are also uniformly bounded for $t\geq 0$.
\end{definition}

A well known result from systems theory which applies to networks which do not have unstable hidden modes shall be our basis of mathematical characterization of BSBR stable networks in terms of their network functions. We shall omit a detailed proof of this result which follows from well known LTI systems theory given in the modern texts \cite{vids} while a single port characterization of stability is discussed in \cite[chapter 4]{chdk}.

\begin{lemma}
A multiport network without unstable hidden modes and represented by (\ref{networkfunction}) is BSBR stable iff the network function matrix $H(s)$ has no poles in the closed right half complex plane denoted as RHP. 
\end{lemma}

\subsection{Source series and source parallel connections of multiport networks}
One of the important aspects of the stabilization problem of multiport networks is that the response of the network is not affected by a compensation unless it is connected at the ports in a specific way. The externally connected network called \emph{compensator} affects the original network only when at least one port of both networks is connected in what we shall describe below as source series form. Hence the stabilization problem and subsequently the robust stabilization problem is meaningful when the networks have ports connected in this fashion. We define two types of connections at a port. Let two networks be $N$ and $N_{c}$ and let $u_{s}$, $u_{cs}$ denote the sources of these networks respectively at an identified port on each network and $y_{r}$, $y_{cr}$ denote the responses of these sources. Following interconnections are possible between the networks only when 1) the source types at any port where interconnection is to be made are the same, i.e. either both are voltage sources or both are current source and 2) the port connections obey the Kirchhoff's law at the cutset of the two port terminals. We assume that the port Kirchhoff's laws are not violated by the interconnection.
\begin{enumerate}
    \item The networks are said to be in \emph{source parallel} connection at the ports if the two sources are equal to a common source $u_{s}$ applied to the connected network at the port while the response at the port is the algebraic sum $y_{r}+y_{cr}$. When all the ports of $N$ and $N_c$ under a fixed indexing have same source type and are connected this way at all the ports of respective index we denote the interconnected network as $p(N,N_{c})$. 
    
    \item The networks are said to be in \emph{source series} connection at the ports if the source applied to the connected network at the port is the algebraic sum $u_{s}+u_{cs}$ while the response at the port is the common response $y_{r}=y_{cr}$. When all the ports of $N$ and $N_c$ under a fixed indexing have same source type and are connected this way at all the ports of respective index we denote the interconnected network as $s(N,N_{c})$.  
\end{enumerate}

An important point is that the responses of the individual networks remain unaltered when all the port connections are source parallel. Hence if the compensating network $N_{c}$ is required to affect response of the given network $N$ then the interconnection has to be source series type at least for one port. Hence our formulation of the stabilization problem is defined only for source series connection $s(N,N_{c})$. We shall make following assumption in all our connections at ports.

\begin{assumption}\label{assumonintcon}\emph{
A multiport network $N$ which is under consideration for stabilization is connected by a compensating network $N_{c}$ which has the same number of ports with same type of sources at each of the indices and is connected in source series form at all the ports at corresponding indices. Thus our assumption excludes problems in which number of source series connections at ports are less than the total number of ports.
}
\end{assumption}

\begin{problem}Let $N$ be a given multiport network with a fixed indexing of its ports and let the type of sources at these ports be fixed. Consider that a compensating network $N_{c}$ with same number of ports and same type of sources at the same indices, is connected to $N$ in source series form at all the ports to form the interconnection $s(N,N_{c})$. Then what is the characterization of the set of all multiport networks $N_c$ such that 1) The interconnected network $s(N,N_{c})$ is BSBR stable with respect to independent sources at the interconnections and responses of the sources at each port and 2) the interconnection $s(\tilde{N},N_{c})$ is also BSBR stable for all $\tilde{N}$ in a sufficiently small neighborhood of $N$. When networks $N_{c}$ exist they are called \emph{stabilizing compensators} of $N$. 
\end{problem}

\begin{remark}\emph{
This question is analogous to that of the question in feedback system theory, "what are all possible stabilizing feedback controllers of a LTI plant?" However unlike the system theory formulation at a fixed parameter the stabilization problem above also requires stability of interconnection to be satisfied even for sufficiently small perturbations of the original network. An important difference of the stabilization problem from the feedback stabilization problem is as follows. In the feedback theory the notion of \emph{internal stability} of a feedback system in terms of the signal flow graph originally defined in \cite{dlms} is BIBO stability with respect to two external inputs and two external outputs in the feedback loop.  Such an injection of external inputs is possible in feedback systems because the output and input are at two different places in the signal flow graph. On the other hand in the case of network interconnection, external source at a port and the response due to it are at the same port. Hence the two input two output notion of stability is not even meaningful for port interconnection. Hence we made an alternative formulation of the stabilization problem as above.
}
\end{remark}

In Systems Theory the question of characterizing all stabilizing controllers led to landmark new developments in recent decades \cite{vids, doft, zhdg} broadly known as factorization approach based $H_{\infty}$ control. However for network interconnections such a problem appears never to have been posed formally as far as known to the authors although stability of a single port connection of active networks has been well known \cite[chapter 11]{chdk}. 

In \cite{msm1} the above problem is solved for multiport networks by first establishing conditions for BSBR stability of multiport networks in terms of their hybrid functions and prove existence of $N_{c}$ given $N$ when there are no unstable hidden modes in $N$. The coprime fractional framework turns out to be highly appropriate for formulating and solving this problem. The results of the paper are presented in subsequent section to build background for describing the robust network stabilization problem.



\section{Stabilization and Robust stabilization by network interconnection}
\label{robuststproforsec}
As stated in the introduction, the stabilization problem for interconnection of networks $N$ and $N_{c}$ is defined only for source series connection of ports. We thus consider the two networks $N$, $N_{c}$ with following assumptions.
\begin{enumerate}
\item Assumption 1. The two networks $N$, $N_{c}$ have same number of ports and same type of sources at the ports of same indices hence the source series interconnection $s(N,N_{c})$ is defined.
\item Assumption 2. The assumption stated in introduction section that all ports of $N$ and $N_c$ of same indices are connected in source series form. 
\item Assumption 3. Both $N$, $N_{c}$ have no unstable hidden modes and have driving point functions given respectively by (\ref{networkfunction}) for $N$ and
\begin{equation}
\label{eq4}
Y_{cr}= H_c U_{cs}
\end{equation}
for $N_{c}$.
\item Assumption 4. Both $H$ and $H_{c}$ are nonsingular, proper and have nonsingular and proper inverses $H^{-1}$, $H_{c}^{-1}$.
\end{enumerate} 
In the interconnection $s(N,N_{c})$ the sources satisfy the equation
\begin{equation}
\hat{U}_{s}=U_{s}+U_{cs}
\end{equation}
while the common response vector $\hat{Y}_{r}$ of the two networks at the ports is given by the following equation.
\begin{equation}
\hat{Y}_{r}=HU_{s}=H_{c}U_{cs}
\end{equation}
Hence, the source vectors reflected on ports of each network are given by the following equations.
\begin{equation}
U_{s}=H^{-1}\hat{Y}_{r}, \ \ \ \ \ \mbox{   }U_{cs}=H_{c}^{-1}\hat{Y}_{r}
\end{equation}
Therefore, for the interconnected network, the source response relationship is given by the following equation.
\begin{equation}\label{TotalResponserelationmult}
\hat{Y}_{r}=(H^{-1}+H_{c}^{-1})^{-1}\hat{U}_{s}
\end{equation}
which is the hybrid representation of the interconnected network. It can be observed that the interconnected network is BSBR stable iff the hybrid matrix of interconnection $(H^{-1}+H_{c}^{-1})^{-1}$ is in $M(\stabS)$. The stabilization problem is defined under the restriction that the hybrid matrices of the interconnection arising from all $\tilde{H}$, due to uncertainties and perturbations, in sufficiently small neighborhood of $H$ are also stable. The stabilization problem now translates to the following.

\begin{problem}[Multi-port Hybrid Stabilization]
Given a multi-port hybrid matrix function $H$ of an LTI network, find all hybrid network function matrices $H_{c}$ of the compensating network such that the source series interconnection $s(N,N_{c})$ satisfies
\begin{enumerate}
\item $\hat{H}=(H^{-1}+H_{c}^{-1})^{-1}$ is in $M(\stabS)$.
\item $\tilde{\hat{H}}=(\tilde{H}^{-1}+H_{c}^{-1})^{-1}$ is in $M(\stabS)$ for all $\tilde{H}$ in a neighbourhood of $H$.
\end{enumerate}
The matrix functions $H_{c}$ shall be called \emph{stabilizing hybrid compensators} of $H$.
\end{problem}

This problem can be solved using doubly coprime fractional (DCF) representation 
and the notion of neighborhood of a network function. For a comprehensive formulation of the multi-port stabilization we resort to the matrix case of coprime factorization theory over the $\stabS$ developed in \cite{vids}. The doubly coprime representation of $H$ is given as
\begin{enumerate}
\item $H$ is expressed by right and left fractions $H=N_{r}D_{r}^{-1}=D_{l}^{-1}N_{l}$ where $N_{r},D_{r},N_{l},D_{l}$ are matrices over $M(\stabS)$, $D_{r},D_{l}$ are square and have no zeros at infinity,
\item There exist matrices $X_{l},Y_{l}$ and $X_{r},Y_{r}$ in $M(\stabS)$ which satisfy the following equation.
\beq\label{dcf}
\left[\begin{array}{rr}
 X_{l} & Y_{l}\\
 D_{l} & -N_{l} 
\end{array}\right]
\left[\begin{array}{rr}
 N_{r} & Y_{r}\\
 D_{r} & -X_{r}
\end{array}\right]=
\left[\begin{array}{rr}
 I & 0\\
 0 & I
\end{array}\right]
\eeq
\end{enumerate} 

We describe the doubly coprime fractional representation of a compensating network with hybrid network function $H_c$ by the respective matrices of fractions and identities by $N_{cr},D_{cr},N_{cl},D_{cl}$ and $X_{cr},Y_{cr},X_{cl},Y_{cl}$. It is also useful to recall that a square matrix $U$ in $M(\stabS)$ is called \emph{unimodular} if $U^{-1}$ also belongs to $M(\stabS)$. This is true iff $\det U$ is a unit or an invertible element of $\stabS$. 

Next, an open neighbourhood of $H$ is specified in graph topology in terms of neighbourhoods of the doubly coprime fractional representation of $H$. Any $\tilde{H}$ in a neighbourhood of $H$ is specified by a doubly coprime fractional representation 
with fractions $\tilde{H}=\tilde{N}_{r}\tilde{D}_{r}^{-1}=\tilde{D}_{l}^{-1}\tilde{N}_{l}$ and matrices $\tilde{X}_{l},\tilde{Y}_{l}$ and $\tilde{X}_{r},\tilde{Y}_{r}$ in $M(\stabS)$ satisfying the identities as given in equation (\ref{dcf}) in which the fractions $\tilde{N}_{r},\tilde{D}_{r}$,
$\tilde{D}_{l},\tilde{N}_{l}$ are in respective neighbourhoods of the fractions of $H$.

In terms of the doubly coprime fractional (DCF) representation and the notion of neighbourhoods we have the preliminary.

\begin{theorem}Consider the port interconnection $s(N,N_{c})$. Then $H_c$ is a stabilizing compensator of $H$ iff for a given doubly coprime fractions as above of $H$ there exist a doubly coprime fractions of $H_{c}$ that satisfy the following equation.
\begin{equation}\label{dcfrelationwithTc}
\left[\begin{array}{rr}
 D_{cl} & N_{cl}\\
 D_{l} & -N_{l} 
\end{array}\right]
\left[\begin{array}{rr}
 N_{r} & N_{cr}\\
 D_{r} & -D_{cr}
\end{array}\right]=
\left[\begin{array}{rr}
 I & 0\\
 0 & I
\end{array}\right]
\end{equation}
\end{theorem}


The structure of stabilizing compensators $H_c$ now follows from the equation (\ref{dcfrelationwithTc}) in terms of the DCF of $H$ as follows.

\begin{cor}
Given a DCF (\ref{dcf}) of $H$ the set of all stabilizing compensators $H_c$ are given by any of the following alternative formulae.
\beq\label{Tcparametrization}
\begin{array}{rcl}
H_{c} & = & (X_{l}-QD_{l})^{-1}(Y_{l}+QN_{l})\\
H_{c} & = & (Y_{r}+N_{r}Q)(X_{r}-D_{r}Q)^{-1}
\end{array}
\eeq
for all $Q$ in $M(\stabS)$ such that functions $\det (X_{l}-QD_{l})$ and $\det (X_{r}-D_{r}Q)$ have no zero at infinity.
\end{cor}

\subsection{Robust Stabilization Problem}
Consider $N_{l}$ and $D_{l}$ to be matrices over $M(\stabS)$ at nominal parameter values of a given network. Due to uncertainties, perturbations and model errors, these matrices differ from their nominal values such that $\tilde N_{l}$ and $\tilde {D_{l}}$ form the neighborhood of $N_{l}$ and $D_{l}$. Solving the robust stabilization problem involves finding a compensating network, $H_{c}$ with right coprime fractions $N_{cr},D_{cr}$ such that the interconnection remains stable even in the neighborhood of $N_{l}$ and $D_{l}$.
Let $\Delta$ be the matrix representing the neighborhood of $N_{l}$ and $D_{l}$ and is as given below.
\begin{equation}\Delta = 
\left[\begin{array}{rr}
\tilde N_{l}-N_{l} & \ \ \tilde D_{l}-D_{l}
 \end{array}\right] =
 \left[\begin{array}{rr}
\Delta_{N} & \ \ \Delta_{D}
 \end{array}\right]
\end{equation}
Thus $\tilde N_{l}$ and $\tilde D_{l}$ can be expressed as,
\begin{align}
\tilde N_{l}&=N_{l}+\Delta_{N} \notag \\
\tilde D_{l}&=D_{l}+\Delta_{D}
\end{align} 
Now, consider
\begin{align*}
 \left[\begin{array}{rr}
\tilde N_{l} & \ \ \tilde D_{l}
 \end{array}\right] 
 \left[\begin{array}{rr}
 D_{cr} \\
 N_{cr}
 \end{array}\right] 
 &= 
 \left[\begin{array}{rr}
N_{l}+\Delta_{N} & \ \ D_{l}+\Delta_{D}
 \end{array}\right] 
 \left[\begin{array}{rr}
 D_{cr} \\
 N_{cr}
 \end{array}\right] \\
 &= \Big(
 \left[\begin{array}{rr}
N_{l} & \ \ D_{l}
 \end{array}\right]  
 +
 \left[\begin{array}{rr}
\Delta_{N} & \ \ \Delta_{D}
 \end{array}\right] \Big)
  \left[\begin{array}{rr}
 D_{cr} \\
 N_{cr}
 \end{array}\right]
 \end{align*}
But $N_{l}D_{cr}+D_{l}N_{cr}=I$ and $\Delta =\left[\begin{array}{rr}
\Delta_{N} & \ \ \Delta_{D}
 \end{array}\right]$.
Therefore, the above equation can be simplified as shown below.
\begin{equation}
\left[\begin{array}{rr}
\tilde N_{l} & \ \ \tilde D_{l}
\end{array}\right] 
\left[\begin{array}{rr}
D_{cr} \\
N_{cr}
\end{array}\right] 
=
I + \Delta \left[\begin{array}{rr}
D_{cr} \\
N_{cr}
\end{array}\right] 
\end{equation} 

With nominal parameter values, the identity $N_{l}D_{cr}+D_{l}N_{cr}=I$ is satisfied but with perturbations and disturbances 
$
 I + \Delta \left[\begin{array}{rr}
 D_{cr} \\
 N_{cr}
 \end{array}\right]
 $
 must also belong to $M(\stabS)$ so as to achieve robust stabilization. For this, following equation must be satisfied.
 \begin{equation}
 \Big \Vert \Delta \left[\begin{array}{rr}
 D_{cr} \\
 N_{cr}
 \end{array}\right] \Big \Vert_{\infty}<1
 \end{equation}
 
$\Delta$ represents various matrices based on uncertainties which should satisfy $\Vert \Delta (j\omega)) \Vert \leq R(j\omega) \ \ \forall \ \ \omega \in [0,\infty)$ where $R(j\omega)$ represents the uppermost bound on $\Delta$. In worst case, $\Delta$ can be considered to be equal to $R$ being one of the possible functions. Thus, the robust stabilization problem now reduces to the following. 
\begin{equation}
\Big \Vert R \left[\begin{array}{rr}
D_{cr} \\
N_{cr}
\end{array}\right] \Big \Vert_{\infty}<1
\end{equation}
 
But $D_{cr}$ and $N_{cr}$, from corollary 
are respectively given as $(X_{r}-D_{r}Q)$ and $(Y_{r}+N_{r}Q)$ in terms of a free parameter matrix $Q$. Thus, the robust stabilization problem can further be simplified as below.
 \begin{equation}
 \label{robuststabi1}
 \Big \Vert R \left[\begin{array}{rr}
X_{r}-D_{r}Q \\
Y_{r}+N_{r}Q
 \end{array}\right] \Big \Vert_{\infty}<1
 \end{equation}
Simplifying equation (\ref{robuststabi1}) further gives,
\begin{equation}
\label{robuststabfor}
\Big \Vert R 
\left[\begin{array}{rr}
X_{r} \\
Y_{r}
\end{array}\right] - R
\left[\begin{array}{rr}
D_{r} \\
-N_{r}
\end{array}\right] Q \Big \Vert_{\infty}<1 
\equiv
\Vert T_{1}-T_{2}Q\Vert_{\infty}<1
\end{equation}
Hence, the robust stabilization problem can be solved iff there a solution to the following optimization problem,
\beq
\min_{Q\in M(\stabS)} \Vert T_{1}-T_{2}Q\Vert_{\infty}<1
\eeq
This formulation of the robust stabilization problem is mathematically identical to the well known formulation in feedback control theory and also has well known approaches for solution \cite{doft,fran}.
\section{Solution to robust stabilization problem}
\label{solutionrobuststab}
The solution to robust stabilization problem is explored in this section. Here, $T_{i}$' s are matrix valued functions. The problem is lot harder to solve as compared to when they are scalar valued. This problem is typically a model matching problem which computes a matrix $Q$ in $\stabS$ so as to minimize $\Vert T_{1}-T_{2}Q\Vert$ from known matrices $T_{1}$ and $T_{2}$. \\
Let $\gamma$ denote the infimal model matching error as given below.
\begin{equation}
\gamma := \textrm{inf}\{\Vert T_{1}-T_{2}Q\Vert_{\infty} : Q \in \stabS \}
\end{equation} 
It can be solved by computing the upper bound $\beta$ for $\gamma$ such that $\beta-\gamma$ is less than a pre-specified tolerance and then $Q$ in $\stabS$ can be computed which satisfies the given robust stabilization problem as below.
\beq
\min_{Q\in M(\stabS)} \Vert T_{1}-T_{2}Q\Vert_{\infty}\leq \beta
\eeq
$Q$ obtained after solving this problem may not be optimal but it will be as near optimality as desired.

For solving such a problem in scalar case, one of the approaches as described in \cite{fran}, is to reduce it to a Nehari problem which approximates $RL_{\infty}$ matrix\footnote{A $RL_{\infty}$ matrix is a real rational matrix which is proper and has no poles on imaginary axis.} by a matrix belonging to $\stabS$. The problem can be simplified by bringing in an inner-outer factorization of $T_{2}$ as $T_{2}=T_{2i}T_{2o}$. For $Q$ in $\stabS$, following simplification can be carried out.  
\begin{align*}
\Vert T_{1}-T_{2}Q\Vert &= \Vert T_{1}-T_{2i}T_{2o}Q\Vert  \\
&=\Vert T_{2i}(T_{2i}^{-1} T_{1}-T_{2o}Q\Vert  \\ 
&=\Vert T_{2i}^{-1}T_{1}-T_{2o}Q\Vert  \\ 
&=\Vert R-X \Vert
\end{align*}

This simplification is possible due to the property of the inner factor $\vert T_{2i}(j\omega)\vert=1$. The inner factor does not affect the infinity norm and can be taken out from the norm.    

The Nehari problem can be posed as:- Given $R$ in $RL_{\infty}$ with dist $(R,\stabS)<1$, find all $X's$ in $\stabS$ such that   $\Vert R-X \Vert_{\infty}\leq 1$. Only some of these are closest to $R$ that satisfy  
$\Vert R-X \Vert_{\infty}=$ dist$(R,\stabS)$.

It is required to find the distance from an $L_{\infty}$ matrix $R$ to $\stabS$. In systemic terms, a given unstable transfer function in $L_{\infty}$ norm is to be approximated by a stable one. Nehari's theorem is an elegant solution to this problem. A lower bound for the distance can be easily obtained. Fix $X$ in $\stabS$ then $\Vert R-X \Vert_{\infty}= \Vert \Gamma_{R}\Vert $ where $\Gamma_{R}$ is the Hankel operator. $ \Vert \Gamma_{R}\Vert $ is a lower bound for the distance $R$ to $\stabS$. Nehari's theorem states that there exists a closest matrix $X$ in $\stabS$ to a given matrix $R$ in $L_{\infty}$ and $\Vert R-X \Vert =\Vert \Gamma_{R}\Vert$. 

In general, there are many $X$'s nearest $R$. The time domain interpretation of Nehari's theorem states that the distance from a given noncausal system to the nearest causal linear and time invariant one equals the norm of the Hankel operator. Alternatively, the norm of the Hankel operator is a measure of non-causality.

To see the systematic development for the solution to robust stabilization problem in matrix case, first an inner-outer factorization of $T_{2}$ can be carried out as $T_{2}=U_{i}U_{o}$ where $U_{i}$ is inner and $U_{o}$ is outer. Define $RL_{\infty}$ matrix $Y$ as $Y:=I-U_{i}U_{i}^{*}T_{1}$ where $U_{i}^{*}$ is complex conjugate transpose of the matrix $U_{i}$. For defining the spectral factorization of a matrix, consider a square matrix $G(s)$ which satisfies the following properties. 
\begin{align}
G,G^{-1} \ &\in RL_{\infty} \notag \\
G^{*} &= G \notag \\
G(\infty) &>0 \notag 
\end{align}
Such a matrix has pole and zero symmetry about the imaginary axis. $G$ can be factored as $G=G_{-}^{*} G$ where $G_{-},G_{-}^{-1} \in \stabS$. This is called a spectral factorization of $G$ and $G_{-}$ is a spectral factor.  

If $\beta$ is a real number greater than $\Vert Y \Vert_{\infty}$, then the matrix $\beta^{2}-Y^{*}Y$ has a spectral factor $Y_{o}$ where $Y^{*}$ is complex conjugate transpose of the matrix $Y$. Let $R$ be defined as a 
$RL_{\infty}$ matrix as $R:=U_{i}^{*}T_{1}Y_{o}^{-1}$. Thus, $R$ depends on $\beta$. It is now required to find a closest matrix $X$ in $\stabS$ to a given matrix $R$ in $L_{\infty}$.

Following are the two preliminary technical facts that need to be considered before stating the main theorem for finding $X$ in $\stabS$ in the form of lemmas as found in \cite {fran} without proofs and used in the derivation of the main theorem.
\begin{lem}\label{lema1}\emph{
If $U$ is an inner matrix and $E$ is an $RL_{\infty}$ matrix 
\begin{equation*}
 E:= \Big \Vert  \left[\begin{array}{rr}
 U^{*} \\
 I-UU^{*}
 \end{array}\right] \Big \Vert_{\infty}
 \end{equation*}
then $\Vert EG \Vert_{\infty}=\Vert G \Vert_{\infty} $ for all matrices $G$ in $RL_{\infty}$ matrix.}
\end{lem}
\begin{lem}\label{lema2}\emph{
If $F$ and $G$ are $RL_{\infty}$ matrices with equal number of columns and if
\begin{equation*}
 \Big \Vert  \left[\begin{array}{rr}
 F \\
 G
 \end{array}\right] \Big \Vert_{\infty} < \beta
 \end{equation*}
then $\Vert G \Vert_{\infty} < \beta$ and $\Vert FG_{o}^{-1} \Vert_{\infty} < 1$ where $G_{o}$ is a spectral factor of $\beta^{2}-G^{*}G$.  Conversely, if $\Vert G \Vert_{\infty} < \beta$ holds and $\Vert FG_{o}^{-1} \Vert_{\infty} < 1$, then 
\begin{equation*}
 \Big \Vert  \left[\begin{array}{rr}
 F \\
 G
 \end{array}\right] \Big \Vert_{\infty} \leq \beta
 \end{equation*}
}
\end{lem}

With these preliminary results, the main theorem can now be stated which gives $X$ in $\stabS$ such that $\Vert R-X \Vert_{\infty}<1$ from which $Q$ in $\stabS$ can be obtained. The proof for the theorem as given in \cite{fran} indicates the approach for finding the solution to robust stabilization problem.

\begin{theorem}
\begin{enumerate}
    \item $\gamma := \textnormal{inf}\{\beta: \Vert Y \Vert_{\infty} < \beta, \ \textnormal{dist}(R,\stabS)<1\}$
    \item \textnormal{Suppose} $\beta > \gamma,\ Q,X \in \stabS$,\ $\Vert R-X \Vert_{\infty}\leq 1$ \textnormal{and} $X=U_{o}QY_{o}^{-1}$
\end{enumerate}
\textnormal{Then} $\Vert T_{1}-T_{2}Q\Vert_{\infty} \leq \beta $
\end{theorem}

\begin{proof}
Part ($1$) of this theorem provides a method for computing an upper bound $\beta$ for $\gamma$ while part ($2$) yields a procedure for computation of a nearly optimal $Q$.

Let $\beta_{inf} := \textnormal{inf}\{\beta: \Vert Y \Vert_{\infty} < \beta, \ \textnormal{dist}(R,\stabS)<1\}$.
Choose $\epsilon>0$ and then choose $\beta$ such that $\gamma + \epsilon > \beta > \gamma$. Then, there exists $Q$ in $\stabS$ such that $\Vert T_{1}-T_{2}Q\Vert_{\infty} \leq \beta$.
From lemma \ref{lema1},
\begin{align*}
 \Big \Vert  \left[\begin{array}{rr}
 U_{i}^{*} \\
 I-U_{i}U_{i}^{*}
 \end{array}\right] (T_{1}-T_{2}Q)\Big \Vert &< \beta \\
 \Big \Vert  \left[\begin{array}{rr}
 U_{i}^{*} \\
 I-U_{i}U_{i}^{*}
 \end{array}\right] T_{1} -
 \left[\begin{array}{rr}
 U_{i}^{*} \\
 I-U_{i}U_{i}^{*}
 \end{array}\right] 
 T_{2}Q \Big \Vert &< \beta
 \end{align*}
Now, consider
\begin{equation*}
 \left[\begin{array}{rr}
 U_{i}^{*} \\
 I-U_{i}U_{i}^{*}
 \end{array}\right] 
 T_{2}=  
 \left[\begin{array}{rr}
 U_{i}^{*} \\
 I-U_{i}U_{i}^{*}
 \end{array}\right] 
 U_{i}U_{o} =
 \left[\begin{array}{rr}
 U_{o} \\
 0
 \end{array}\right] 
  \end{equation*}
 Therefore, 
 \begin{align*}
 \Big \Vert  \left[\begin{array}{rr}
 U_{i}^{*} \\
 I-U_{i}U_{i}^{*}
 \end{array}\right] T_{1} -
 \left[\begin{array}{rr}
 U_{o} \\
 0
 \end{array}\right] Q
 \Big \Vert_{\infty} &< \beta \\
 \therefore \Big \Vert  \left[\begin{array}{rr}
 U_{i}^{*}T_{1}-U_{o}Q \\
 (I-U_{i}U_{i}^{*})T_{1}
 \end{array}\right] 
 \Big \Vert_{\infty} &< \beta \\
 \therefore \Big \Vert  \left[\begin{array}{rr}
 U_{i}^{*}T_{1}-U_{o}Q \\
 Y
 \end{array}\right] 
 \Big \Vert _{\infty} &< \beta
 \end{align*} 
 From lemma (\ref{lema2}), 
 \begin{align}
 \Vert Y \Vert_{\infty} & < \beta \\
 \label{spectfactineq}
 \Vert U_{i}^{*}T_{1}Y_{o}^{-1}-U_{o}QY_{o}^{-1} \Vert_{\infty} & < 1
 \end{align}
 
Inequality (\ref{spectfactineq}) implies dist $(R,\ U_{o} \ \stabS  Y_{o}^{-1})<1$. But, $U_{o}$ is right invertible in $\stabS$ and $Y_{o}$ is invertible in $\stabS$. Therefore, $U_{o} \ \stabS  Y_{o}^{-1} = \stabS$ and hence dist $(R,\stabS) <1$. 

From $\Vert Y \Vert_{\infty} < \beta$, dist $(R,\stabS) <1$ and definition of $\beta_{inf}$, it can be concluded that $\beta_{inf} \leq \beta$. Thus, $\beta_{inf} < \gamma+\epsilon$. Since $\epsilon$ is arbitrary, $\beta_{inf} \leq \gamma$.   
 
Now for the reverse inequality, select $\epsilon>0$ and then choose $\beta$ such that $\beta_{inf}+\epsilon>\beta>\beta_{inf}$. Then, $\Vert Y \Vert_{\infty} < \beta$, dist $(R,\stabS) <1$ hold. So, $\Vert U_{i}^{*}T_{1}Y_{o}^{-1}-U_{o}QY_{o}^{-1} \Vert_{\infty}<1$ also holds for some $Q$ in $\stabS$. From lemma (\ref{lema2}),
 \begin{align*}
 \Big \Vert  \left[\begin{array}{rr}
 U_{i}^{*}T_{1}-U_{o}Q \\
 Y
 \end{array}\right] 
 \Big \Vert _{\infty} \leq \beta
 \end{align*} 
 This results in $\Vert T_{1}-T_{2}Q\Vert_{\infty} \leq \beta$. Therefore, $\gamma \leq \beta < \beta_{inf}$ so $\gamma \leq \beta_{inf}$.
 \end{proof}
 
It is required to find $X$ such that $\Vert R-X\Vert_{\infty} \leq 1$ which gives $Q$ such that $\Vert T_{1}-T_{2}Q\Vert_{\infty} \leq \beta$ and the optimization problem is solved. The algorithm as found in \cite{fran} for computing nearly optimal $Q$ is as given below.
\begin{enumerate}
\item [\textbf {Step 1}] Compute $Y$ and  $\Vert Y \Vert_{\infty}$ where $Y:=(I-U_{i}U_{i}^{*})T_{1}$. $U_{i}$ is the inner factor of $T_{2}$ 
\item [\textbf {Step 2}] Find an upper bound $\gamma_{1}$ for $\gamma$. The simplest bound is $\gamma_{1}=\Vert T_{1}
\Vert_{\infty}$.
\item [\textbf {Step 3}] Select a trial value for $\beta$ in the interval ($\Vert Y \Vert_{\infty}, \gamma_{1}$). Binary search can be used to iterate on $\beta$.
\item [\textbf {Step 4}] Compute $R$ and $\Vert \Gamma_{R} \Vert$. Then $\Vert \Gamma_{R} \Vert < 1$ iff $\gamma < \beta$. Increase or decrease the value of $\beta$ accordingly and return to \textbf {Step 3}. With a sufficiently accurate upper bound for $\gamma$ is obtained, continue to \textbf {Step 5}.
\item [\textbf {Step 5}] Find a matrix $X$ in $\stabS$ such that $\Vert R-X \Vert_{\infty}\leq 1$.
\item [\textbf {Step 6}] Solve for $Q$ in $\stabS$ as $X=U_{o}QY_{o}^{-1}$. It is easier to solve this equation for $Q$ if $U_{o}$ is square than when it is not square.   
\end{enumerate}
The illustrative example solved in the next section shows the design of a compensator for a given network such that the interconnected network is robustly stable under parametric variations.
\section{Illustrative example for robust stabilization problem}
\label{robuststabprobsec}
In this section, an illustrative example is solved to show the design of a robust compensator for a given network in order to achieve the robust stability under parametric variations. To obtain such a compensating network, it is required to solve the robust stabilization problem as posed in section \ref{solutionrobuststab}. 

Consider the single port network as shown in the Fig.\ref{fig03:name} excited by a current source. The impedance for this network can be found using Kirchhoff's laws. 
\begin{figure}[!ht]
\centering
\includegraphics[width=9cm,height=4cm]{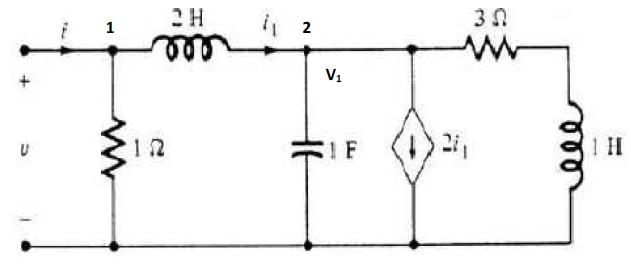}
\caption{Single port network example for robust stabilization }\label{fig03:name}
\end{figure}
Applying KCL at node 1,
\begin{equation}
\label{keqn1}
i=\frac{v}{R_{1}}+\frac{v-v_{1}}{sL_{1}}=\frac{v}{R_{1}}+\frac{v}{sL_{1}}-\frac{v_{1}}{sL_{1}}=\Big(\frac{1}{R_{1}}+\frac{1}{sL_{1}}\Big)v-\Big(\frac{1}{sL_{1}}\Big)v_{1}
\end{equation}
Applying KCL at node 2,
\begin{align}
\label{keqn2}
\frac{v-v_{1}}{sL_{1}}&=\frac{v_{1}}{\Big(\frac{1}{sC_{1}}\Big)}+2i_{1}+\frac{v_{1}}{sL_{2}+R_{2}} \notag \\
&= sC_{1}v_{1} + 2\Big(\frac{v-v_{1}}{sL_{1}}\Big)+\frac{v_{1}}{sL_{2}+R_{2}} \notag \\
\therefore -\frac{v-v_{1}}{sL_{1}}&=\Big(sC_{1}+\frac{1}{sL_{2}+R_{2}}\Big)v_{1}\notag \\
\therefore -\frac{v}{sL_{1}}&=\Big(sC_{1}+\frac{1}{sL_{2}+R_{2}}-\frac{1}{sL_{1}}\Big)v_{1} \notag \\
\therefore v_{1}&=\Big[\frac{-(sL_{2}+R_{2})}{L_{1}L_{2}C_{1}s^{3}+L_{1}C_{1}R_{2}s^{2}+(L_{1}-L_{2})s-R_{2}}\Big]v 
\end{align}
Substituting $v_{1}$ obtained in terms of $v$ from equation (\ref{keqn2}) in equation (\ref{keqn1}), following equation can be obtained.
\begin{align}
\label{kequn3}
i&=\Big(\frac{1}{R_{1}}+\frac{1}{sL_{1}}\Big)v-\Big(\frac{1}{sL_{1}}\Big)\Big[\frac{-(sL_{2}+R_{2})}{L_{1}L_{2}C_{1}s^{3}+L_{1}C_{1}R_{2}s^{2}+(L_{1}-L_{2})s-R_{2}}\Big]v  \notag \\
\therefore i&= \Big[\frac{sL_{1}+R_{1}}{sL_{1}R_{1}}+\frac{sL_{2}+R_{2}}{sL_{1}(L_{1}L_{2}C_{1}s^{3}+L_{1}C_{1}R_{2}s^{2}+(L_{1}-L_{2})s-R_{2})}\Big]v 
\end{align}
The admittance $Y$, of the single port network after simplifying equation (\ref{kequn3}), is as given below.
\begin{align}
\label{kequn4}
Y(s)&=\frac{I(s)}{V(s)}=\frac{L_{1}L_{2}C_{1}s^3+(L_{1}C_{1}R_{2}+L_{2}C_{1}R_{1})s^2+(L_{1}-L_{2}+C_{1}R_{1}R_{2})s+(R_{1}-R_{2})}{R_{1}L_{1}L_{2}C_{1}s^3+R_{1}R_{2}L_{1}C_{1}s^2+R_{1}(L_{1}-L_{2})s-R_{1}R_{2}}  
\end{align}
The impedance $Z$, of the single port network can be obtained by taking the reciprocal of the equation (\ref{kequn4}) which is as given below.  
\begin{align}
Z(s)&=\frac{V(s)}{I(s)}=\frac{R_{1}L_{1}L_{2}C_{1}s^3+R_{1}R_{2}L_{1}C_{1}s^2+R_{1}(L_{1}-L_{2})s-R_{1}R_{2}}{L_{1}L_{2}C_{1}s^3+(L_{1}C_{1}R_{2}+L_{2}C_{1}R_{1})s^2+(L_{1}-L_{2}+C_{1}R_{1}R_{2})s+(R_{1}-R_{2})}
\end{align}
At nominal parameter values $R_{1}=1 \Omega, R_{2}=3 \Omega, L_{1}=2H, L_{2}=1H $ and $C_{1}=1F$, following are the expressions for $Z(s)$ and $Y(s)$.
\begin{align}
Y(s)&=\frac{2s^{3}+7s^2+4s-2}{2s^{3}+6s^2+s-3}=\frac{(s+2.5707)(s+1.2424)(s-0.3131)}{(s+2.5811)(s+1)(s-0.5811)} \\
Z(s)&=\frac{2s^{3}+6s^2+s-3}{2s^{3}+7s^2+4s-2}=\frac{(s+2.5811)(s+1)(s-0.5811)}{(s+2.5707)(s+1.2424)(s-0.3131)}
\end{align}

The pole-zero plot of $Z(s)$ at nominal parameter values is as shown in the Fig.\ref{fig06b:name}. The left half plane pole (zero) of $Y(s)$ ($Z(s)$) at $s=-2.5811$ and zero (pole) of $Y(s)$ ($Z(s)$) at $s=-2.5707$ are located approximately at the same location. Hence, they can be cancelled so as to obtain the minimal realization for $Y(s)$ and $Z(s)$) which are, respectively given by the following equations. 
\begin{figure}[!ht]
\centering
\includegraphics[width=15.0cm,height=10cm]{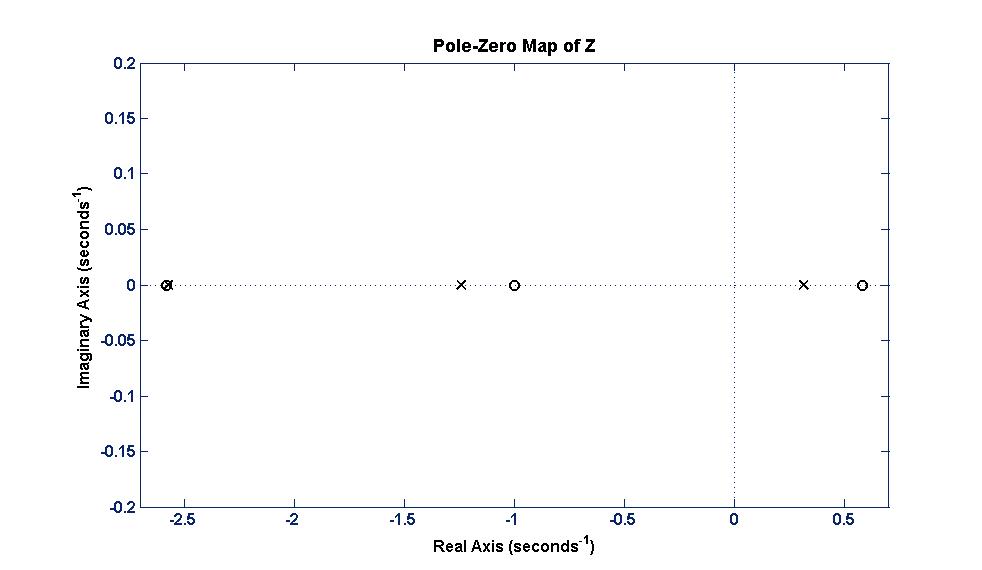}
\caption{Pole-zero plot of $Z$}\label{fig06b:name}
\end{figure}
\begin{align}
Y(s)&=\frac{s^2+0.9293s-0.389}{s^2+0.4189s-0.5811}=\frac{(s+1.2424)(s-0.3131)}{(s+1)(s-0.5811)} \\
\label{impedance}
Z(s)&=\frac{s^2+0.4189s-0.5811}{s^2+0.9293s-0.389}=\frac{(s+1)(s-0.5811)}{(s+1.2424)(s-0.3131)}
\end{align}
A model of the compensating network shall be determined which under parametric variations of the given network ensures the stability of the interconnected network.

The given one port network at nominal parameter values is neither open circuit stable as the pole at $s=0.3131$ of its impedance function $Z(s)$ lies in open RHP nor short circuit stable as the pole at $s=0.5811$ of its admittance function $Y(s)$ lies in open RHP. A compensating network can be designed so as to make the resultant network open circuit stable under the variations in parameters. For this, consider that $\pm 5\%$ variations occur in the parameters of the given single port network.  

The stabilizing compensator $Y_{c}=Z_{c}^{-1}$ is a parallel admittance and set of all such stabilizing $Z_{c}$ has parametrization as $Z_{c}=(Y+QN)(X-QD)^{-1}$ where $Q$ is an arbitrary element of $\stabS$ such that $X-QD\neq0$. It is required to find $Q$ such that the interconnected network is robustally stable under parametric variations.

In this case, the port relation is $V=ZI$. Hence, $H=Z$, $H^{-1}=Y=1/Z$. $H$ is represented by coprime fractions, $Z=ND^{-1}$ where $N,D$ are in $\stabS$. Using the coprime factors,$N, D$ of $Z$, it is possible to obtain two rational functions $X$ and $Y$ such that $NX+DY=1$ holds. The source series compensation is provided by an admittance $Y_{c}=1/Z_{c}$ connected in parallel with $Y$. 
Such a compensator has coprime representation $Z_{c}=N_{c}D_{c}^{-1}$. \\ 
The coprime fractions $N$ and $D$ for $Z$ as given by equation (\ref{impedance}) at nominal parameter values are as below.
\begin{align}
N(s)&=
\frac{s^{2}+0.4189s-0.5811}{s^{2}+3s+2} = \frac{s-0.5811}{s+2} \\
D(s)&=\frac{s^2+0.9293s-0.389}{s^2+3s+2}=
\frac{(s+1.2424)(s-0.3131)}{(s+1)(s+2)}
\end{align}
The solutions to Bezout's identity $XN+YD=1$ give $Y(s)$ and $X(s)$ as below.
\begin{align}
X(s)&=
\frac{-16.9s-20.92}{s^{2}+3s+2} \\
Y(s)&=
\frac{s+20.97}{s+2}
\end{align}

The solution to robust stabilization problem involves finding $R(j\omega)$ which represents an uppermost bound on uncertainties represented by $\Delta$. To find this uppermost bound, it is required to consider the network functions with parametric variations.  Consider $\pm 5\%$ variations occur in the parameters of this network. For each set of variations, $\tilde Z$ and its respective coprime fractions $\tilde N$ and $\tilde D$ needs to be computed. With $N$ and $D$ as coprime fractions at nominal parameter values and $\tilde N$ and $\tilde D$ as coprime fractions at each set of parameter variations, following matrix can be formed.   
\begin{equation}
\label{perturbmatrix}
\left[\begin{array}{rr}
\tilde N-N\\
\tilde D-D  
\end{array}\right]
\end{equation}

There are five parameters and each of these parameters can take $+5\%$ and $-5\%$ apart from its nominal value. Thus, each of these parameters undergoes three possible states ($-5\%$, nominal value and $+5\%$) which forms $3^{5}=243$ possible combinations of the perturbed network impedances and corresponding coprime fractions for each need to be computed.

In order to obtain the uppermost bound, it is required to compute $\infty$ norm of each matrix as given by the equation (\ref{perturbmatrix}) where $\tilde N$ represents the coprime fraction of the perturbed network impedance for each of the parameter variations while $N$ represents the coprime fraction of the network impedance with nominal parameter values. Thus, $3^{5}-1=243-1=242$ matrices are required to be considered for computing the $\infty$ norm at various frequency points wherein the robust stabilization problem needs to be solved. Calculations for one of the variations in the parameters out of different possible combinations of parametric variations is shown here as an illustration.\\
In the network function $Z$, let each of the parameter is undergoing a variation of say $+5\%$ from its nominal value so that  $\tilde R_{1}=1.05 \Omega, \tilde R_{2}=3.15 \Omega, \tilde L_{1}=2.1H, \tilde L_{2}=1.05H $ and $\tilde C_{1}=1.05F$, following expression for perturbed impedance, $\tilde Z(s)$ is obtained.
\begin{align}
\tilde Z&=\frac{2.4310s^{3}+7.2930s^2+1.1025s-3.3075}{2.3152s^{3}+8.1033s^2+4.5228s-2.1} =\frac{1.05(s+2.631)(s+0.9268)(s-0.5579)}{(s+2.624)(s+1.171)(s-0.2951)} \notag \\ 
\label{tildeimpedance}
\therefore \tilde Z&=\frac{1.05(s+0.9268)(s-0.5579)}{(s+1.171)(s-0.2951)}
\end{align} \\
The coprime fractions $\tilde N$ and $\tilde D$ for $\tilde Z$ as given by equation (\ref{tildeimpedance}) for $+5\%$ parameter variation are as below.
\begin{align}
\tilde N &=
\frac{1.05s^{2}+0.3874s-0.543}{s^{2}+3s+2} = \frac{1.05(s+0.9268)(s-0.5579)}{(s+1)(s+2)} \\
\tilde D &=\frac{s^{2}+0.8759s-0.3456}{s^{2}+3s+2}=
\frac{(s+1.171)(s-0.2951)}{(s+1)(s+2)}
\end{align}
The solutions to Bezout's identity $\tilde X \tilde N+\tilde Y \tilde D=1$ give $\tilde Y$ and $\tilde X$ as below.
\begin{align}
\tilde X&=
\frac{-17.9s-20.92}{s^{2}+3s+2} = \frac{-17.9(s+1.1687)}{(s+1)(s+2)}\\
\tilde Y&=\frac{s^{2}+23.92s+21.29}{s^{2}+3s+2}=\frac{(s+22.99)(s+0.9259)}{(s+1)(s+2)}
\end{align}
Thus, following equations are obtained.
\begin{align}
\tilde N-N &=\frac{0.05(s^{2}-0.632s+0.765)}{(s+1)(s+2)}\\
\tilde D-D&=\frac{-0.053s+0.0433}{s^{2}+3s+2}=\frac{-0.05343(s-0.817)}{(s+1)(s+2)}
\end{align}
Using above equations, following matrix can be obtained.
\begin{equation}
\label{perturbmatrixexample}
\left[\begin{array}{rr}
\tilde N-N\\
 \tilde D-D  
\end{array}\right]=
\left[\begin{array}{rr}
\frac{0.05(s^{2}-0.632s+0.765)}{(s+1)(s+2)}\\ \\
\frac{-0.05343(s-0.817)}{(s+1)(s+2)}
\end{array}\right]
\end{equation}

In the matrix given by equation (\ref{perturbmatrixexample}), various frequencies over which robust stabilization problem is required to be solved can be substituted and at each frequency, infinity norm of the matrix can be computed. The variation of this infinity norm with frequency can then be plotted. 

For each of the parameter variations, it is required to compute first the matrix as given by equation (\ref {perturbmatrix}) and then its infinity norm for various frequency points. The variation of this infinity norm with frequency for each of the matrix can then be plotted. Such a plot is as shown in the Fig.\ref{fig05:name}.
\begin{figure}[!ht]
\centering
\includegraphics[width=15.8cm,height=9cm]{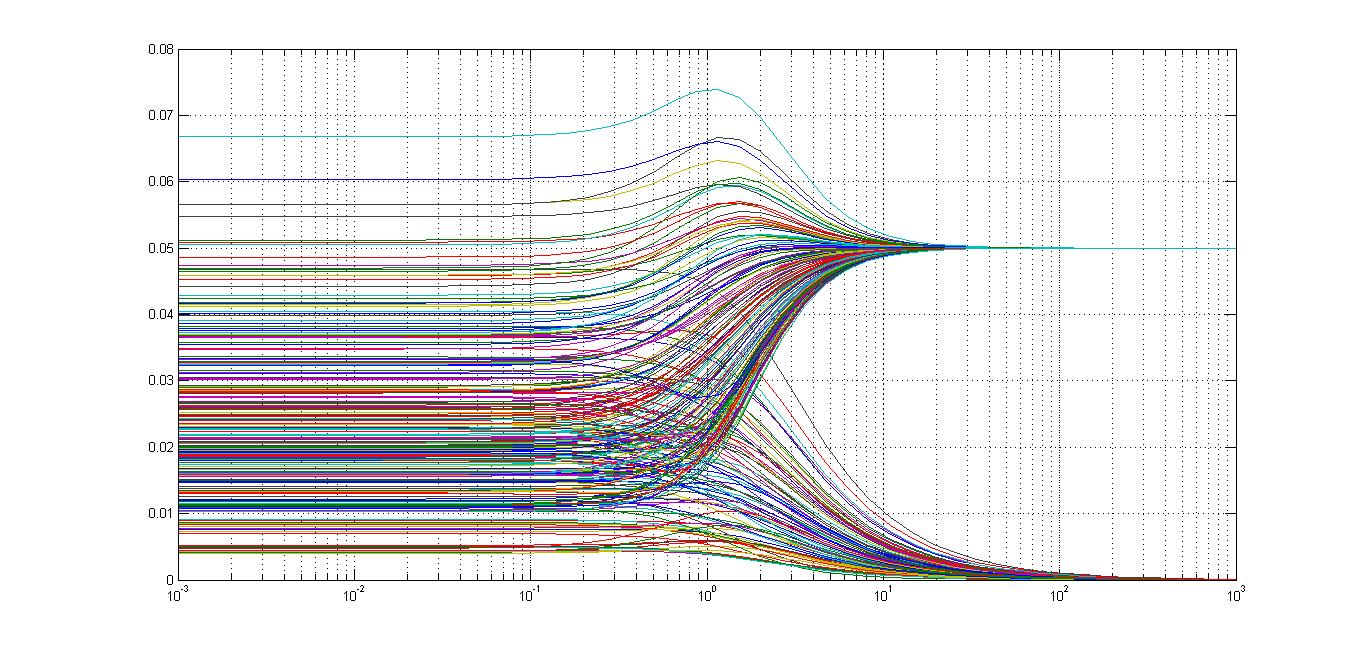}
\caption{Variation of infinity norms of perturbed matrices (absolute value) with frequencies}\label{fig05:name}
\end{figure} 
The same plot with db values plotted on Y-axis is as shown in the Fig.\ref{fig06:name}.
\begin{figure}[!ht]
\centering
\includegraphics[width=16.2cm,height=9.4cm]{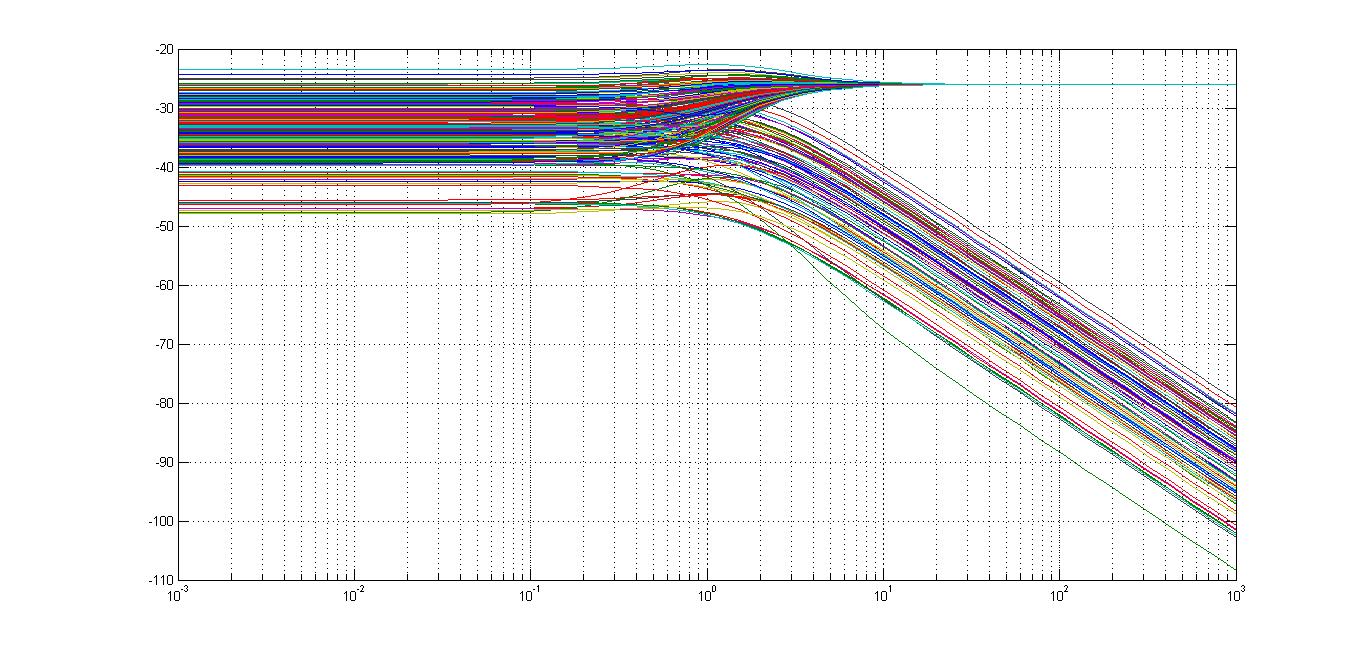}
\caption{Variation of infinity norms of perturbed matrices (in db) with frequencies}\label{fig06:name}
\end{figure} 
From the family of loci, the uppermost bound $R(j\omega)$ can be computed.

All the other functions will lie below this uppermost bound. Such a function $R(s)$ can be found and is given by the following equation with its Bode magnitude plot as shown in the Fig.\ref{fig06c:name}.
\begin{equation}
R(s)=\frac{0.06(s+6)}{s+3.4}    
\end{equation} 
\begin{figure}[!ht]
\centering
\includegraphics[width=16.2cm,height=9.4cm]{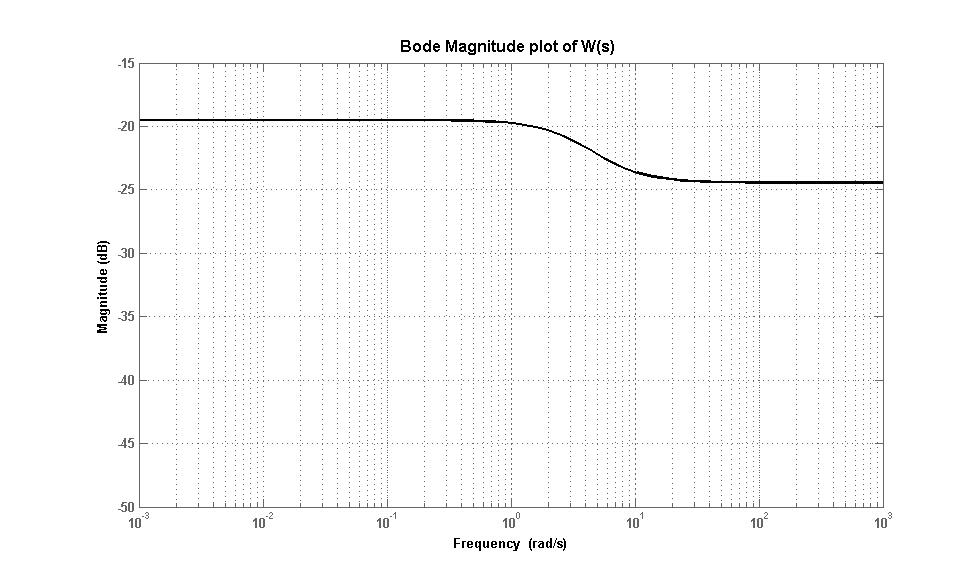}
\caption{Bode magnitude plot of $R(s)$}\label{fig06c:name}
\end{figure} \\
As seen from equation (\ref{robuststabfor}),
\begin{equation}\label{robuststabequnpara1}
T_{1}=R
\left[\begin{array}{rr}
 X_{r}\\
 Y_{r}  
\end{array}\right]=\frac{0.06(s+6)}{s+3.4}
\left[\begin{array}{rr}
\frac{-16.9s-20.92}{s^{2}+3s+2}\\
\frac{s+20.97}{s+2}  
\end{array}\right]=\left[\begin{array}{rr}
\frac{-1.014(s+1.238)(s+6)}{(s+1)(s+2)(s+3.4)}\\
\frac{0.06(s+20.97)(s+6)}{(s+2)(s+3.4)}  
\end{array}\right]
\end{equation}
\begin{equation}\label{robuststabequnpara2}
T_{2}=R
\left[\begin{array}{rr}
 D_{r}\\
 -N_{r}  
\end{array}\right]=
\left[\begin{array}{rr}
\frac{0.06(s+1.242)(s+6)(s-0.3131)}{(s+1)(s+2)(s+3.4)}\\
\frac{-0.06(s-0.5811)(s+6)}{(s+2)(s+3.4)}  
\end{array}\right]
\end{equation}
The robust stabilization problem can now be solved iff there a solution to the following optimization problem
\beq
\min_{Q\in M(\stabS)} \Vert T_{1}-T_{2}Q\Vert_{\infty}<1
\eeq

Converting this problem to Nehari problem which can be posed as for a given $R$ in $RL_{\infty}$ with dist $(R,\stabS)<1$, find all $X's$ in $\stabS$ such that $\Vert R-X \Vert_{\infty}\leq 1$ where $R=U_{i}^{*}T_{1}Y_{o}^{-1}$. Applying the steps in the algorithm from \cite{fran}, nearly optimal $Q$ can be computed. First, the inner-outer factorization of $T_{2}$ can be computed as $T_{2}=T_{2i}T_{2o}=U_{i}U_{0}$.
\begin{align}
U_{i}&=\left[\begin{array}{rr}
\frac{0.70711(s+1.242)(s-0.3131)}{(s+1.141)(s+0.4334)}\\
\frac{-0.70711(s-0.5811)(s+1)}{(s+1.141)(s+0.4334)}  
\end{array}\right]\\
U_{o}&=\frac{0.084853(s+0.4334)(s+1.141)(s+2)(s+3.4)(s+6)}{(s+1)(s+2)^{2}(s+3.4)^{2}}
\end{align}
As discussed in the \textbf{step 1} of the algorithm for solving the robust stabilization problem in matrix case, it is required to compute $RL_{\infty}$ matrix $Y$ which is given by, $Y=(I-U_{i}U_{i}^{*})T_{1}$.
\begin{equation}
Y=\left[\begin{array}{rr}
\frac{0.03(s+1)(s+0.5811)(s-1)(s+6)(s^{2}+4s+4.001)}{(s+1.141)(s-1.141)(s+0.4334)(s-0.4334)(s+2)(s+3.4)}\\
\frac{0.03(s+1)(s+6)(s+0.3131)(s-1.242)(s^{2}+4s+4.001)}{(s+1.141)(s-1.141)(s+0.4334)(s-0.4334)(s+2)(s+3.4)}  
\end{array}\right]   
\end{equation}

From above equation, $\Vert Y \Vert_{\infty}=0.3028$. Next, an upper bound $\gamma_{1}$ for $\gamma$ can be found which is $\gamma_{1}=\Vert T_{1} \Vert_{\infty}=1.5682$. Thus, $\beta$ lies in the interval ($\Vert Y \Vert_{\infty},\gamma_{1})=(0.3028,1.5682)$. For selecting a trial value for $\beta$ in this interval, binary search can be used. 

If $\beta > \Vert Y \Vert_{\infty}$, then the matrix $\beta^{2}-Y^{*}Y$ has a spectral factor $Y_{o}$ where $Y^{*}$ is complex conjugate transpose of the matrix $Y$. Let $R$ be $RL_{\infty}$ matrix defined as $R:=U_{i}^{*}T_{1}Y_{o}^{-1}$. $R$ depends on the value of $\beta$ selected. It is now required to find a closest matrix $X$ in $\stabS$ to a given matrix $R$ in $L_{\infty}$. The table \ref{table:1} below shows the variation of $\Vert\Gamma_{R}\Vert$ for different values of $\beta$.

\begin{table}[!ht]
\centering
\begin{tabular}{|c|c|}
\hline
$\beta$ & $\Vert\Gamma_{R}\Vert$ \\
\hline
0.9355 & 1.3004 \\
1.2518 & 0.9593 \\
1.0937 & 1.1035 \\
1.1728 & 1.0263 \\
1.2123 & 0.9916 \\
\hline
\end{tabular}
\caption{Table showing $\Vert\Gamma_{R}\Vert$ for various values of $\beta$ }
\label{table:1}
\end{table}

For each value of selected $\beta$, $R$ and $\Vert\Gamma_{R}\Vert$ is computed. For $\beta=1.2123, \Vert\Gamma_{R}\Vert=0.9916$ so this value of $\beta$ can be considered for which, ($\beta^{2}-Y^{*}Y$) has a spectral factor $Y_{o}$ which is as given below.
\begin{equation}
Y_{o}=\frac{1.2116(s+3.397)(s+1.14)(s+0.4208)}{(s+0.4334)(s+1.141)(s+3.4)}
\end{equation}
From the spectral factor $Y_{o}$, $R=U_{i}^{*}T_{1}Y_{o}^{-1}$ can be easily computed which is given by the following equation.
\begin{equation}
R=\frac{-0.03501(s+37.99)(s+6)(s+1.141)(s+1.139)(s+0.433)(s+0.4226)(s-1.101)}{(s-1.141)(s-0.4334)(s+1.14)(s+1)(s+2)(s+3.397)(s+0.4208)}
\end{equation}
Since $R$ is a scalar valued function, following theorem as found in \cite{fran} can be applied to find the closest function $X$ in $\stabS$.
\begin{theorem}
The infimal model matching error in model matching error equals  $\Vert \Gamma_{R} \Vert$, the unique optimal $X$ equals $R-\gamma (\frac{f}{g})$ and for the optimal $Q$, $T_{1}-T_{2}Q$ is all pass. 
\end{theorem}

The algorithm to compute $X$, when $R$ is scalar valued, can be found in \cite{fran}. It  is given below for reference and  applicable for this example.
\begin{enumerate}
\item [\textbf {Step 1}] Factor $R$ as $R=R_{1}+R_{2}$ where $R_{1}$ is strictly proper and analytic in Re s $\leq$ 0 while $R_{2}$ belongs to $\stabS$. $R_{1}$ has a minimal state space realization with $A,B,C$ as its matrices and matrix $A$ is antistable. 
\item [\textbf {Step 2}] Solve the following equations for $L_{c}$ and $L_{o}$.
\begin{align}
AL_{c}+L_{c}A^{T}&=BB^{T} \notag \\
A^{T}L_{o}+L_{o}A&=C^{T}C \notag
\end{align}
\item [\textbf {Step 3}] Find the maximum eigenvalue $\lambda^{2}$ of $L_{c}L_{o}$ and a corresponding eigenvector $w$. Defining $v:\lambda^{-1}L_{o}w$ so that from equation $L_{c}L_{o}w=\lambda^{2}w$, following pair of equations can be obtained. Solve the equation for $v$.
\begin{align*}
L_{c}v=\lambda w  \\
L_{o}w=\lambda v
\end{align*}
\item [\textbf {Step 4}] Define the real rational functions $f(s)$ and $g(s)$ as shown below.
\begin{align}
f(s)&:=
\left[\begin{array}{c|c}
A & w \\
\hline
C & 0
\end{array}\right]=C(sI-A)^{-1}w \notag \\
g(s)&:=
\left[\begin{array}{c|c}
-A^{T} & v \\
\hline
B^{T} & 0
\end{array}\right]=B^{T}[sI-(-A^{T})]^{-1}v \notag
\end{align}
\item [\textbf {Step 5}] Set the model matching error $\gamma=\lambda$ so that $X=R-\gamma (\frac{f}{g})$.
\item [\textbf {Step 6}] Solve for $Q$ in $\stabS$ using $X$ as obtained in \textbf{step 5}.    
\end{enumerate}

Using these steps from the algorithm, $X$ and hence $Q$ as a solution to the robust stabilization problem can be computed. From \textbf{step 1}, $R$ can be factored as $R=R_{1}+R_{2}$. $R_{1}$ and $R_{2}$ are respectively given by following equations.
\begin{align}
R_{1}&=\frac{-0.090464(s-1.09)}{(s-0.4334)(s-1.141)} \\
R_{2}&=\frac{-0.035018(s+1.14)(s+1.097)(s+0.4206)(s^{2}+19.11s+121)}{(s+0.4208)(s+1)(s+1.14)(s+2)(s+s+3.397)}
\end{align}
$R_{1}$ has a minimal state space realization with matrix $A$ antistable. The state space realization of $R_{1}$ is as given below.
\begin{align}
A&=\left[\begin{array}{rr}
0.4334 & 0.8104\\
0 & 1.1410 
\end{array}\right] \notag \\
B&= \left[\begin{array}{rr}
0 \\
1 
\end{array}\right]\notag \\
C&= \left[\begin{array}{rr}
0.7331 & -0.9046\\
\end{array}\right]\notag
\end{align}
The Lyaupnov equations given in \textbf{step 2} can be solved so as to obtain the controllability and observability grammians which are as given below.
\begin{align}
L_{c}&=\left[\begin{array}{rr}
0.4218 & -0.2256\\
-0.2256 & 0.4382 
\end{array}\right] \notag \\
L_{o}&= \left[\begin{array}{rr}
0.6201 & -0.7405\\
-0.7405 & 0.8846 
\end{array}\right]\notag 
\end{align}
The largest eigenvalue of $L_{c}L_{o}$ denoted as $\lambda^{2}$ is $0.9833$ and the corresponding eigenvector $w$ such that $L_{c}L_{o}w=\lambda^{2}w$ holds is as given below,
\begin{equation}
w= \left[\begin{array}{rr}
0.6782 \\
-0.7349 
\end{array}\right] \notag
\end{equation}
$v$ which is defined as $\lambda^{-1}L_{o}w$ can be computed by solving the equation $L_{o}w=\lambda v$ and is as given below.
\begin{equation}
v= \left[\begin{array}{rr}
0.9729 \\
-1.1620 
\end{array}\right] \notag
\end{equation}
The real rational functions $f(s)$ and $g(s)$ as defined in \textbf{step 4} can now be computed and are given below.
\begin{align}
f(s)&=\frac{1.162s-1.292}{s^{2}-1.574s+0.4945} \\
g(s)&=\frac{-1.162s-1.292}{s^{2}+1.574s+0.4945}
\end{align}
From \textbf{step 5}, the model matching error $\gamma=\lambda=\sqrt{0.9833}=0.9916$. Using this value of $\lambda$, $X$ can be obtained which is given below.
\begin{equation}
X=R-\gamma \Big(\frac{f}{g}\Big)=\frac{0.95658(s+4.27)(s+1.141)(s+1.141)(s+0.4331)}{(s+1.112)(s+1)(s+2)(s+3.397)}  
\end{equation}
Solving for $Q$ using the equation $X=U_{o}QY_{o}^{-1}$, it is possible obtain $Q$ which is given below.
\begin{equation}
Q=U_{o}^{-1}XY_{o}=\frac{13.6584(s+1.141)(s+0.4208)(s+4.27)}{(s+1.112)(s+6)(s+0.4334)}  
\end{equation}
The compensating network $Z_{c}$ can be expressed in terms of the free parameter $Q$ as below.
\begin{equation}
Z_{c}=\frac{Y+NQ}{X-DQ}  
\end{equation}
Using the expressions for $N,D,X,Y$ and $Q$ already computed, the compensating network $Z_{c}$ for the given network $Z$ can be computed which is given below.
\begin{equation}
\label{compnetwork}
Z_{c}=\frac{-1.0732(s+3.455)(s+1)}{(s+1.231)(s+3.387)}
\end{equation}
The compensating network $Z_{c}$ as given by equation (\ref{compnetwork}) ensures the robust stability of the interconnected network under $\pm 5\%$ variations in the parameters of the given network. 

For this circuit example, it can now be shown that the impedance of the interconnected network $\tilde Z_{T}=(\tilde Z^{-1}+Z_{c}^{-1})^{-1}$ under parametric variations of $\pm 5\%$ is in $\stabS$. Consider the perturbed impedance $\tilde Z$ as given by equation (\ref{tildeimpedance}) with each parameter undergoing a variation of say,$+5\% $, as one of the possible variations. Now, it can be shown that the impedance of the interconnected network $\tilde Z_{T}$ is in $\stabS$.
\begin{align*}
\tilde Z_{T}&=(\tilde Z^{-1}+Z_{c}^{-1})^{-1} \notag \\
&=\Big\{\Big[\frac{1.05(s+0.9268)(s-0.5579)}{(s+1.171)(s-0.2951)}\Big]^{-1} + \Big[\frac{-1.0732(s+3.455)(s+1)}{(s+1.231)(s+3.387)}\Big]^{-1}\Big \}^{-1}\\
\therefore \ \tilde Z_{T}&=\frac{48.5413(s+3.455)(s+1)(s+0.9268)(s-0.5579)}{(s+16.24)(s+2.462)(s+1.462)(s+0.7239)}
\end{align*}
This shows that the perturbed impedance of the interconnected network $\tilde Z_{T}$ for one set of the parameter variations is in $\stabS$. 

Perturbed impedance, $\tilde Z_{T}$ for one set of parameter variations is observed to be in $\stabS$. For other set of the parameter variations, $\tilde Z$ can similarly be computed and it can be shown that the perturbed impedance of the interconnected network $\tilde Z_{T}$ with the compensating network, $Z_{c}$ connected in parallel with $\tilde{Z}$ is also found to be in $\stabS$. The pole-zero plot for $\tilde Z_{T}$ for each set of parametric variations is as shown in the Fig.\ref{fig06a:name} and it can be seen that the all the poles of $\tilde Z_{T}$ lie in the left half plane.
\begin{figure}[!ht]
\centering
\includegraphics[width=15.0cm,height=10cm]{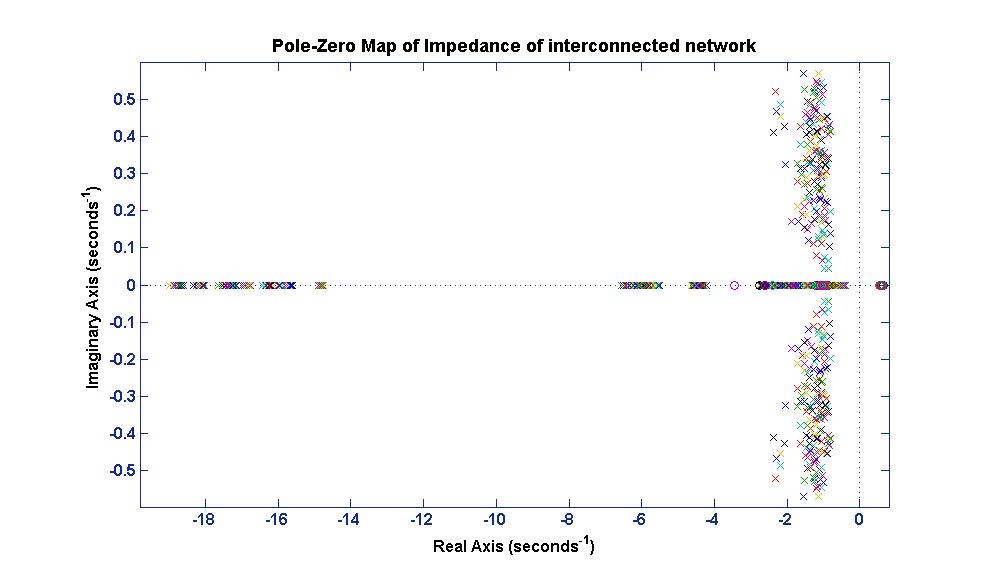}
\caption{Pole-zero plot of $\tilde Z_{T}$ for all parameter variations}\label{fig06a:name}
\end{figure}
Thus, the compensator $Z_{c}$ robustly stabilizes the given network under uncertainties and perturbations due to parameter variations.
\section{Conclusion}
\label{robustconsec} 
Robust stabilization problem in case of active multiport networks is formulated and solved in this chapter. Robustness is an important consideration in multiport interconnections as network undergoes perturbations and uncertainties due to parametric variations, model errors, disturbances and other parasitic effects which are inevitable. It may be possible that the the interconnected network is unable to meet the prescribed performance specifications or even become unstable if robustness is not considered during the design of a compensating network. The solution to robust stabilization problem can be obtained by approaches well known in modern system theory. An illustrative circuit example considered in this chapter shows that the compensating network obtained after solving the robust stabilization problem robustly stabilizes the given network under parametric variations. 



\begin{thebibliography}{xxxx}
\bibitem{msm1}
M.V.Bakshi, V.R.Sule, Maryam Shojaei Baghini, 
\newblock``Systems Theory Approach to Stabilization of Multiport Networks",
\newblock {\em Journal of Control and Systems Engineering}, vol. 5, no. 1, pp. 48-63, 2017.

\bibitem{fran}
B.A.Francis, 
\newblock {Lecture notes in control and information sciences},
\newblock Springer-Verlag, NY, 1987.

\bibitem{kimu}
H. Kimura
\newblock ``Robust stabilizability for a class of transfer functions",
\newblock {\em IEEE Transactions on Automatic Control}, vol. AC-29, no. 9, pp. 788--793, 1984.

\bibitem{waima}
Mao-Da Tong, Wai-Kai Chen, 
\newblock ``Analysis of VLSI robust exponential stability with left coprime factorization,''
\newblock Circuits Systems Signal Processing, vol.17, No.3, pp. 335 -360, 1998.

\bibitem{antm}
Panos J. Antsaklis, Anthony N. Michel,
\newblock {Linear Systems},
\newblock Birkh$\ddot{a}$user Bostan, 2nd corrected printing, 2006.

\bibitem{vids}
Vidyasagar M.,
\newblock {Control system synthesis. A factorization approach},
\newblock Research Studies Press, NY, 1982.

\bibitem{horo}
Horowitz I.M.
\newblock {Synthesis of feedback systems},
\newblock Academic Press, NY, 1963.

\bibitem{doft}
Doyle J. C., B. A. Francis, A. R.Tannenbaum,  
\newblock {Feedback Control Theory},
\newblock Macmillan Publishing Company, NY, 1992.

\bibitem{msm2}
M.V. Bakshi, V.R. Sule, Maryam Shojaei Baghini,
\newblock ``Sensitivity minimization in active networks by port compensation'', 
\newblock{\em International Symposium on control systems, SICE ISCS, Okayama, Japan, March 6-9, 2017}, pp. 1--8.

\bibitem{chdk}
 L.O. Chua, C.A. Desoer, E.S. Kuh,
\newblock {Linear and nonlinear circuits},
\newblock McGraw-Hill Book Company, NY, 1987.

\bibitem{dlms}
C.A. Desoer, Ruey-Wen Liu, John Murray, Richard Saeks,
\newblock ``Feedback system design: the fractional representation approach to analysis and synthesis",
\newblock {\em IEEE Transactions on Automatic Control}, vol. AC-25, no.3, pp. 399--412, 1980.

\bibitem{zhdg}
Zhou K. J., Doyle J. C., K. Glover 
\newblock {Robust and Optimal Control},
\newblock Printice Hall, 1996.










































Zhiping Lin,










































































\end{thebibliography}
\end{document}